\documentclass[12pt]{article}
\RequirePackage[colorlinks,citecolor=blue,urlcolor=blue]{hyperref}

\usepackage{amsmath, amssymb, amsthm, booktabs}
\usepackage{graphicx,psfrag,epsf}
\usepackage{enumerate}
\usepackage[usenames]{color}
\usepackage{natbib}
\usepackage{url} 
\usepackage{algorithm,float} 
\usepackage{algpseudocode}
\usepackage{multirow,enumerate}    

\definecolor{red2}{rgb}{0.7, 0, 0.1}

\newtheorem{lemma}{Lemma}

\newtheorem{theorem}{Theorem}

\newtheorem{remark}{Remark}

\newcommand{\beq}{\begin{equation}}
\newcommand{\eeq}{\end{equation}}
\newcommand{\beas}{\begin{eqnarray*}}
\newcommand{\eeas}{\end{eqnarray*}}
\newcommand{\bea}{\begin{eqnarray}}
\newcommand{\eea}{\end{eqnarray}}
\def\pr{\textsf{P}} 
\def\ep{\textsf{E}}

\def\S{{\Sigma}}

\begin{document}

\def\spacingset#1{\renewcommand{\baselinestretch}%
{#1}\small\normalsize} \spacingset{1}

\title{A Normality Test for High-dimensional Data based on a Nearest Neighbor Approach}

\author{Hao Chen$^{1},$  \ and \ Yin Xia$^{2}$}

\date{}

\footnotetext[1]{Department of Statistics, University of California at Davis. The research of Hao Chen was supported in part by NSF Grants DMS-1513653 and DMS-1848579.}

\footnotetext[2]{Department of Statistics, School of Management, Fudan University. The research of Yin Xia was supported in part by  NSFC Grants 12022103, 11771094, 11690013.}

\maketitle

\vspace{0.2in}
\begin{abstract}
Many statistical methodologies for high-dimensional data assume the population is normal.  Although a few multivariate normality tests have been proposed, to the best of our knowledge, none of them can properly control the type I error when the dimension is larger than the number of observations.
In this work, we propose a novel nonparametric test that utilizes the nearest neighbor information.
The proposed method guarantees the asymptotic type I error control under the high-dimensional setting. 
Simulation studies verify the empirical size performance of the proposed test when the dimension grows with the sample size and at the same time exhibit  a superior power performance of the new test compared with  alternative methods.  We also illustrate our approach through two popularly used data sets in high-dimensional classification and clustering literatures where deviation from the normality assumption may lead to invalid conclusions. 
\end{abstract}



\spacingset{1.45} 
\newpage

\section{Introduction}
\label{sec:intro}

The population normality assumption is widely adopted in many classical statistical analysis (e.g., linear and quadratic discriminant analysis in classification, normal error linear regression models, and the Hotelling $T^2$-test), as well as many recently developed methodologies, such as network inference through Gaussian graphical models \citep{ma2007arabidopsis, yuan2007model, friedman2008sparse, rothman2008sparse, fan2009network, yuan2010high,  liu2013ggm, xia2015}, high-dimensional linear discriminant analysis \citep{bickel2004some, fan2008high,cai2011direct, mai2012direct}, post-selection inference for regression models \citep{berk2013valid, lee2016exact, taylor2018post}, and change-point analysis for high-dimensional data \citep{xie2013sequential, chan2015optimal, wang2018high, liu2019scalable}. 
When the data is univariate, there are many classical tools to check the normality assumption, such as the normal quantile-quantile plot and the Shapiro-Wilk test \citep{shapiro1965analysis}.  However, many of the modern applications involve multivariate or even high-dimensional data and it constantly calls for multivariate normality testing methods with good theoretical performance.  

In this article, we aim to address the following testing problem in the high-dimensional setting with a proper control of type I error. Given a set of observations $X_1,X_2,\dots,X_n\overset{iid}{\sim} F$, where $F$ is a distribution in $\mathbb{R}^d$, one wishes to test
$$H_0: F \text{ is a multivariate Gaussian distribution},$$
versus the alternative hypothesis
$$H_a: F \text{ is not a multivariate Gaussian distribution}.$$

In the literature, there have been a good number of methods proposed to test the normality of multivariate data.  For example, \cite{mardia1970measures} considered two statistics to measure the multivariate skewness and kurtosis separately, and constructed two tests for the normality of the data by using each of these two statistics;  Bonferroni correction can be applied to unify these two tests.  More recently, \cite{doornik2008omnibus} proposed a way to combine the two test statistics effectively.  In another line, \cite{royston1983some} generalized the Shapiro-Wilk test to the multivariate setting by applying the Shapiro-Wilk test to each of the coordinates and then combining the test statistics from all coordinates, while \cite{fattorini1986remarks} tried to find the projection direction where the data is most non-normal and then applied the Shapiro-Wilk test to the projected data. Later, \cite{zhou2014powerful} combined these two approaches by considering the statistics from both random projections as well as the original coordinates. 
{ In a related work, \cite{villasenor2009generalization} proposed a multivariate Shapiro--Wilk's test based on the transformed test statistics standardized by the sample mean and covariance matrix.}
In addition, {there is a series of work that test normality through the characteristic function}
\citep{baringhaus1988consistent, henze1990class, henze1997new}.  
Besides those methods, there is also another work that extends the Friedman-Rafsky test \citep{friedman1979multivariate}, a nonparametric two-sample test, to a multivariate normality test \citep{smith1988test}.  Those aforementioned methods provide useful tools for testing multivariate normality assumption for the conventional low-dimensional data. 

We illustrate in Table \ref{tab:ex_size} the empirical size for some of the representative existing tests: ``Skewness" (the test based on the measure of multivariate skewness in \cite{mardia1970measures}), ``Kurtosis" (the test based on the measure of multivariate kurtosis in \cite{mardia1970measures}), ``Bonferroni" (the method combining the tests based on multivariate skewness and kurtosis through the Bonferroni correction), ``Ep" (an effective way of combining the multivariate skewness and kurtosis in \cite{doornik2008omnibus}), ``Royston" (generalized Shapiro-Wilk test in \cite{royston1983some}), ``HZ" (the test based on the characteristic function proposed in \cite{henze1990class}), 
{ ``mvSW'' (the multivariate Shapiro--Wilk's test proposed in \cite{villasenor2009generalization})}, 
and ``eFR" (extended Friedman-Rafsky test in \cite{smith1988test}).  { In particular, the multivariate Shapiro--Wilk's test requires smaller dimension than the sample size, and the extended Friedman-Rafsky test requires an estimate of the variance of the distribution while there is a lack of discussions on such estimations in their paper.  In the table, 
we use ``mvSW$_0$'' and ``eFR$_0$" to respectively represent the test proposed in \cite{villasenor2009generalization} and the extended Friedman-Rafsky test that are based on the sample covariance matrix, and 
use ``mvSW'' and ``eFR" to  respectively represent the tests that are based on a newly developed covariance matrix estimation method, the adaptive thresholding approach proposed in \cite{cai2011direct}.
 We observe from the table that, except for the improved tests  ``mvSW'', and ``eFR'', all other existing methods are either not applicable to the cases when the dimension is larger than the sample size, i.e., $d>n$, or cannot control the type I error well when the dimension is high.}  

\begin{table}[htbp]
\caption{Empirical size (estimated from 1,000 trials) of the tests at 0.05 significance level. Data are generated from the standard multivariate Gaussian distribution with $n=100$. The numbers in the table that are larger than 0.1 are bolded (\textbf{cannot} control the size well). The test Ep is not applicable when $d>n$ and the test ``mvSW$_0$'' is not applicable when $d\geq n$.
\label{tab:ex_size}} 
\begin{center}

	\begin{tabular}{ccccccccc}
	\toprule
	$d$ & 5 & 10 & 20 & 50 & 80 & 90 & 100 & 200 \\ \midrule
	Skewness & 0.035 & 0.039 & 0.014 & 0 & 0 & 0  & \textbf{0.114} & \textbf{0.384}   \\ \midrule
	Kurtosis & 0.041 & 0.071 & \textbf{0.254} & \textbf{0.999} & \textbf{1} & \textbf{1} & \textbf{0.950} & \textbf{0.998} \\ \midrule
	Bonferroni & 0.029 & 0.040 & \textbf{0.158} & \textbf{0.994} & \textbf{0.943} & \textbf{1} & \textbf{1} & \textbf{0.997} \\ \midrule
	Ep & 0.053 & 0.059 & 0.046 & 0.044 & 0.047 & 0.040 & \textbf{0.141} & --  \\ \midrule
	Royston & 0.073 & 0.092 & 0.080 & \textbf{0.137} & \textbf{0.129} & \textbf{0.164} & \textbf{0.168} & \textbf{0.245} \\ \midrule
	HZ & 0.048 & 0.051 & 0.051 & 0 & \textbf{1} & \textbf{1} &  \textbf{1} & \textbf{1} \\ \midrule
	mvSW$_0$ & 0.056 & 0.057 & 0.038 & 0.052 & 0.042 & 0.045 & -- & -- \\ \midrule	
	mvSW$^{ a}$ & 0.051 & 0.057 & 0.042 & 0.052 & 0.043 & 0.046 & 0.045 & 0.051  \\ \midrule		
	eFR$_0$ & 0.056 & 0.041 & 0.048 & 0.081 & \textbf{0.153} & \textbf{0.145} & \textbf{0.161} & 0.088 \\ \midrule	
	eFR$^{ b}$ & 0.045 & 0.046 & 0.048 & 0.041 & 0.038 & 0.038 & 0.044 & 0.042 \\
	\bottomrule      
	\end{tabular}
	
\end{center}
\footnotesize{$^{ a}$ The improved multivariate Shapiro--Wilk's test applying to the transformed statistics standardized by the adaptive thresholding covariance estimators in \cite{cai2011direct}.}

\footnotesize{$^{ b}$ The improved extended Friedman-Rafsky test based on the adaptive thresholding covariance estimators.}

\end{table}

The extended Friedman-Rafsky test is based on an edge-count two-sample test proposed in \cite{friedman1979multivariate}.  Due to the curse of dimensionality, it was shown in a recent work, \cite{chen2017new}, that the edge-count two-sample test would suffer from low or even trivial power under some commonly appeared high-dimensional alternatives with typical sample sizes (ranging from hundreds to millions).  The same problem also exists in the extended Friedman-Rafsky test for testing normality in the high-dimensional setting. { Furthermore, the extended Friedman-Rafsky test can no longer properly control the type I error when the dimension is much larger than the sample size, and similarly for the improved multivariate Shapiro--Wilk's test.} We refer the details to the size and power comparisons in Section \ref{sec:simulation}. 

In this paper, we take into consideration the findings in \cite{chen2017new} and propose a novel nonparametric multivariate normality testing procedure based on nearest neighbor information.  Through extensive simulation studies, we observe that the new test has good performance on the type I error control, even when the dimension of the data is larger than the number of observations.  It also exhibits much higher power than ``mvSW'' and ``eFR''  under the high-dimensional setting.  Moreover, we provide theoretical results in controlling the type I error for the new test  when the dimension grows with the sample size. As far as we know, there is a paucity of systematic and theory-guaranteed hypothesis testing solutions developed for such type of problems in the high-dimensional setting, and our proposal offers a timely response. We also apply our test respectively to two data sets, a popularly used lung cancer data set in the linear discriminant analysis literatures \citep{fan2008high,cai2011direct}  where normality is a key assumption, and a colon cancer data set that was used in high-dimensional clustering literature \citep{jin2016influential} where the data are assumed to follow the normal assumption. The testing results provide useful prerequisites for such analyses that are based on the normality assumption. 

The rest of the paper is organized as follows.  In Section \ref{sec:meth}, we propose a new nonparametric procedure to test the normality of the high-dimensional data and introduce the theoretical properties of the new approach.  The performance of the proposed method is examined through simulation studies in Section \ref{sec:simulation} and the method is applied to two data sets in Section \ref{sec:application}.  Section \ref{sec:discussion} discusses a related statistic, possible extensions of the current proposal, and some sensitivity analyses.  The main theorem is proved in Section \ref{sec:thmproof} with technical lemmas collected and proved in Section \ref{sec:lemma}. 

\section{Method and Theory}
\label{sec:meth}
We propose in this section a novel nonparametric algorithm to test the normality of the high-dimensional data. We start with the intuition of the proposed method, and then study the error control of the new approach based on the asymptotic equivalence of two events for searching the nearest neighbors under the null hypothesis.

\subsection{Intuition} 
A key fact of the Gaussian distribution is that it is completely determined by its mean and variance.  Suppose that the mean $(\mu)$ and covariance matrix $(\Sigma)$ of the distribution $F$ are known, then testing whether $F$ is a multivariate Gaussian distribution is the same as testing whether $F=G$, where $G=\mathcal{N}_d(\mu,\Sigma)$. { For this purpose, one may consider goodness-of-fit tests, such as \cite{bartoszynski1997multidimensional} and the approach proposed in \cite{liu2016kernelized} for high-dimensional data.  We could also generate a new set of observations $Y_1,Y_2,\dots,Y_n\overset{iid}{\sim}G$, and apply the two-sample tests, such as \cite{jurevckova2012nonparametric} and \cite{marozzi2015multivariate} and the graph-based two-sample tests \citep{friedman1979multivariate, chen2017new, chen2018weighted}, to examine $F=G$ for arbitrary dimensions.}


However, in practice, the parameters $\mu$ and $\Sigma$ are unknown in general.  To compromise, we use the mean $(\mu_x)$ and covariance matrix $(\Sigma_x)$ estimated from the set of observations $\{X_1,X_2,\dots,X_n\}$ as substitutes.  We could again generate a new set of observations $Y_1,Y_2,\dots,Y_n\overset{iid}{\sim}G_x=\mathcal{N}_d(\mu_x,\Sigma_x)$, but unfortunately, now the original testing problem is no longer equivalent to testing whether $F=G_x$.  

To address this issue, we use the same combination of $\mu_x$ and $\Sigma_x$ to generate another set of independent observations $X_1^*,X_2^*,\dots,X_n^*\overset{iid}{\sim}G_x=\mathcal{N}_d(\mu_x,\Sigma_x)$. Then we estimate the mean and covariance matrix of these new observations and denote them by $\mu_{x^*}$ and $\Sigma_{x^*}$, respectively.  Based on them, we further generate a new set of independent observations from the normal distribution with mean $\mu_{x^*}$ and covariance matrix $\Sigma_{x^*}$, i.e., $Y_1^*,Y_2^*,\dots,Y_n^*\overset{iid}{\sim}\mathcal{N}_d(\mu_{x^*},\Sigma_{x^*})$.  Intuitively, if the null hypothesis $H_0$ is true, i.e., the original distribution $F$ is multivariate Gaussian, then the relationship between $\{X_1,X_2,\dots,X_n\}$ and $\{Y_1,Y_2,\dots,Y_n\}$ would be similar to that of $\{X_1^*,X_2^*,\dots,X_n^*\}$ and $\{Y_1^*,Y_2^*,\dots,Y_n^*\}$.  Henceforth, we shall test whether these two relationships are similar enough to decide whether $F$ is close enough to a Gaussian distribution.  

In \cite{smith1988test}, the Friedman-Rafsky's two-sample test was used for this purpose.  Unfortunately, as will be shown later in Section \ref{sec:simulation}, this test was unable to properly control the type I error when the dimension is growing with the number of observations. 

In order to guarantee the error control in the high-dimensional setting, we use the nearest neighbor information in this article.  To be specific, we pool  $\{X_1,X_2,\dots,X_n\}$ and $\{Y_1,Y_2,\dots,Y_n\}$ together, and for each observation, we find its nearest neighbor, which is defined under the Euclidean distance in the current paper.  Similarly, we pool  $\{X_1^*,X_2^*,\dots,X_n^*\}$ and $\{Y_1^*,Y_2^*,\dots,Y_n^*\}$ together, and again find the nearest neighbor for each observation.   

{Nearest neighbor information has been employed in hypothesis testing that can be applied to high-dimensional data \citep{schilling1986multivariate, henze1988multivariate, chen2015graph, chen2019sequential}.  However, in these work, nearest neighbors were used for two-sample testing, while in contrast, we only have one sample at the beginning of the current setup and then generate a second sample that depends on the original one.  Hence, we need to develop a completely different set of technical tools to investigate the theoretical properties of the current construction.}   Let $YY$ be the event that an observation in $\{Y_1,Y_2,\dots,Y_n\}$ finds its nearest neighbor in $\{Y_1,Y_2,\dots,Y_n\}$, and let $Y^*Y^*$ be the event that an observation in $\{Y_1^*,Y_2^*,\dots,Y_n^*\}$ finds its nearest neighbor in $\{Y_1^*,Y_2^*,\dots,Y_n^*\}$.  We will show below in Theorem \ref{thm:asym} that the events $Y^*Y^*$ and $YY$ are asymptotic equivalent under some suitable conditions. As a result, we can estimate the empirical distribution of the test statistic based on $YY$ through the distribution of the statistic associated with $Y^*Y^*$. Consequently, the type I error of the proposed approach can be properly controlled at some pre-specified significance level.

\subsection{Theorem on asymptotic equivalence} 
Before studying the main result on the asymptotic equivalence between two events of searching nearest neighbors, we first introduce some notation. 
Denote by $\lambda_{\min}(\Sigma)$ and $\lambda_{\max}(\Sigma)$ the smallest and largest eigenvalues of $\Sigma$.
For two sequences of real numbers $\{a_n\}$ and $\{b_n\}$, denote by $a_n=O(b_n)$ if there exist constants $C>c>0$ such that $c|b_n| \le |a_{n}| \leq C|b_{n}|$ for all sufficiently large $n$. 
 We also remark here that, when $d=1$ or $d=2$, the aforementioned univariate and conventional multivariate methods in the introduction can be easily applied to test the normality assumption, and we shall focus in our work on the cases when the dimension $d$ is larger than 2.

We next introduce two assumptions. 
\begin{enumerate}[({A}1)]
	\item \label{C1}
	The eigenvalues of $\Sigma$ satisfy $C_1\leq \lambda_{\min}(\Sigma)\leq \lambda_{\max}(\Sigma)\leq C_2$ for some constants $C_1, C_2>0$.
	
	\item\label{C2}
	There exists an estimator of $\mu$ such that $\|\mu_x-\mu\|_2 \leq O_\pr(1)$, and an estimator of $\S$ such that  $\|\Sigma_x-\Sigma\|_2 = o_{\pr}(n^{-\frac{1}{d}-\frac{(2+a) \log d + \kappa}{2 \log n}})$ with $\kappa = 1-\frac{1}{d}\log |\Sigma| - \log 2$ and $a=\left\{\begin{array}{ll}
	0 & \text{ if } d\log d\leq \log n \\
	\frac{ \log n}{\xi_{d,n} d \log d} & \text{ if } d\log d> \log n \text{ and } d=o(\log n)  \\
	1/\epsilon_d & \text{ otherwise,}
\end{array} \right.$ where $1\ll \xi_{d,n}=o(\log n/d)$ and  $1\ll\epsilon_d=o(\log d)$. 
\end{enumerate}

Under the above two conditions, Theorem \ref{thm:asym} studies the asymptotic equivalence between the events $YY$ and $Y^*Y^*$ under the null hypothesis, which in turn guarantees the type I error control of the proposed method.
\begin{theorem}\label{thm:asym}
Assume (A\ref{C1}) and (A\ref{C2}). Then it follows that, under $H_0$, as $n\rightarrow\infty$,
$$\pr(YY)-\pr(Y^*Y^*) \rightarrow 0.$$
\end{theorem} 
The proof of the theorem is provided in Section \ref{sec:thmproof}.

\begin{remark}\label{remark1}
Assumption (A\ref{C1}) is mild and is widely used in the high-dimensional literature \citep{bickel2008regularized, rothman2008sparse, yuan2010high, cai2014two}.
Assumption (A\ref{C2}) implies the relationship between the dimension $d$ and the sample size $n$. Specifically, $\|\mu_x-\mu\|_2 \leq O_\pr(1)$ can be easily satisfied when $d=O(n^\gamma), \gamma\leq 1$.
For the condition $\|\Sigma_x-\Sigma\|_2 = o_{\pr}(n^{-\frac{1}{d}-\frac{(2+a)\log d + \kappa}{2 \log n}})$, when $d\geq 3$ and $d=O(n^\gamma), \gamma<1/2$, it can be satisfied by many estimators under some regularity conditions. For example, if we apply the adaptive thresholding estimator in \cite{cai2011direct}, and assume that $\Sigma$ is $s_0$ sparse in the sense that there are at most $s_0$ nonzero entries in each row of $\Sigma$, then we have
$\|\Sigma_x-\Sigma\|_2=O_{\pr}(s_0\sqrt{\log d/n}).$  
So the condition holds if $s_0=o(n^{\frac{1}{2}-\xi-\frac{1}{d}})$ for some $\xi>(1+\frac{a}{2})\gamma$, where $a$ is either equal or tending to zero as defined in detail in (A\ref{C2}). 
When $d=O(n^\gamma), \gamma\geq 1/2$, simulation results show that the conclusion holds well when $d>n, d=O(n)$.  There is potential to relax the condition on $\|\Sigma_x-\Sigma\|_2$ in the theorem.  In the current proof, we made big relaxations from Equation \eqref{eq:diff} to \eqref{eq:diff2} and from Equation \eqref{eq:diffy} to \eqref{eq:diffy2} (see Section \ref{sec:thmproof}).  More careful examinations could lead to tighter conditions.  This requires non-trivial efforts and we save it for future work. 
\end{remark}

\begin{remark}\label{remark2}
 The theory based on nearest neighbor information in the high-dimensional setting has so far received little attention in the literature. We provide in this paper a novel proof for the asymptotic equivalence on two events of searching the nearest neighbors and it is among the first methods that utilizes such nonparametric information and in the mean while guarantees the asymptotic type I error control.
\end{remark}

\subsection{Algorithm and theoretical error control} \label{sec:algo}
Based on Theorem \ref{thm:asym}, we could adopt the following algorithm to test the multivariate normality of the data. To be specific, because of the asymptotic equivalence between the events $YY$ and $Y^*Y^*$, we repeatedly generate the data from the multivariate normal distribution with estimated mean and covariance matrix, and use the empirical distribution of the test statistics based on $Y^*Y^*$ to approximate the empirical distribution of the test statistic based on $YY$ under the null hypothesis.

Denote by $r(YY)$ the percent of $Y$'s that find their nearest neighbors in $\{Y_1,\dots,Y_n\}$, and $r(Y^*Y^*)$ is defined similarly for $Y^*$'s. Let $m(r(Y^*Y^*))$ be the average of the $r(Y^*Y^*)$'s from Step 3 of the algorithm. We then propose a nonparametric normality test based on nearest neighbor information as the following.

\begin{algorithm}[H]
\caption{}
\label{algo1}
\begin{algorithmic}[1]
\item Generate $Y_1, \dots, Y_n \overset{iid}{\sim} \mathcal{N}_d(\mu_x, \Sigma_x)$, calculate $r(YY)$.

\item \label{step2} Generate $X_1^*, \dots, X_n^* \overset{iid}{\sim} \mathcal{N}_d(\mu_x, \Sigma_x)$, estimate its mean $\mu_{x^*}$ and covariance matrix $\Sigma_{x^*}$.  
Generate $Y_1^*, \dots, Y_n^* \overset{iid}{\sim} \mathcal{N}_d(\mu_{x^*}, \Sigma_{x^*})$, calculate $r(Y^*Y^*)$.

\item Repeat Step 2 for $B$ times to get an estimate of the empirical distribution of $r(YY)$ under $H_0$.

\item Compute the two-sided {\it sampling $p$-value}, $p(YY)$, i.e., the percentage of $|r(Y^*Y^*)- m(r(Y^*Y^*))|$ (out of $B$) that are larger than or equal to $|r(YY)- m(r(Y^*Y^*))|$, where $|\cdot|$ is the absolute value.

\item For a given significance level $0<\alpha<1$, define $\Psi_\alpha = I\{p(YY) \leq \alpha\}$. Reject the null hypothesis whenever $\Psi_\alpha = 1$.
\end{algorithmic}
\end{algorithm}




Note that Algorithm \ref{algo1} is a simplified version of a more sophisticated algorithm that generates $n$ independent sets of $\{Y_1, \dots, Y_n\}$ and $\{Y_1^*, \dots, Y_n^*\}$ in Steps 1 and 2, with $r(YY)$ and $r(Y^*Y^*)$ respectively representing the percent of $Y_1$'s and $Y_1^*$'s that find their nearest neighbors in their corresponding sets. The resulting test $\Psi_\alpha^*$ of such an algorithm guarantees the type I error control based on Theorem \ref{thm:asym}, i.e., $\pr(\text{Type I error}) = \pr_{H_0}(\Psi_\alpha^* = 1) \rightarrow \alpha, \mbox{ as } n, B\rightarrow \infty$. However, this algorithm is computationally much more expensive and it has asymptotically the same size performance as Algorithm \ref{algo1}, and hence we mainly focus on Algorithm \ref{algo1} in the current article.



In the implementation, we use the sample mean to obtain $\mu_x$ and $\mu_{x^*}$ and use the adaptive thresholding method in \cite{cai2011direct} to compute $\Sigma_x$ and $\Sigma_{x^*}$.  For the selection of $B$, the empirical distribution can be more precisely estimated when $B$ is larger. We choose $B=500$ in the implementation and it provides well error control as shown in Section \ref{sec:simulation}.
{ It is worthwhile to note that, for faster and easier implementation of the method, the $p$-value $p(YY)$ we obtain in Algorithm \ref{algo1} is random, and we hence call it ``sampling $p$-value''. To improve the power performance of the method, we can further increase the number of such sampling procedure, and the details are discussed in Section \ref{sec:L10}.}




\section{Simulation Studies}
\label{sec:simulation}
We analyze in this section the numerical performance of the newly developed algorithm. As we studied in the introduction, the existing methods ``Skewness'', ``Kurtosis'', ``Bonferroni'', ``Ep'', ``Royston'' and ``mvSW$_0$'' all suffer from serious size distortion or are not applicable when the dimension is relatively large. { We thus consider in this section the size and power comparisons of our approach with the method ``eFR'', in which the covariances are estimated by the adaptive thresholding method in \cite{cai2011direct}, the multivariate Shapiro--Wilk's test (mvSW) proposed in \cite{villasenor2009generalization} applying to the transformed statistics standardized by the adaptive thresholding covariance estimators, as well as the Fisher's test (Fisher) by combining the $p$-values for each dimension of the aforementioned adaptive thresholding covariance standardized Shapiro--Wilk's test.} { As suggested by \cite{cai2011direct}, we use the fivefold cross validation to choose the tuning parameter. Once we obtain an estimator $\hat \Sigma^*$,  we let $\hat \Sigma=(\hat \Sigma^*+\delta I)/(1+\delta)$ with $\delta=\max\{-\lambda_{\min}(\S^*),0\}+0.05$ to guarantee the positive definiteness of the estimated covariance matrix.}

The following matrix models are used to generate the data. {Note that Model 3 considers the nearly singular scenario where the condition number is} {around 80 in typical simulation runs when $d=100$.}
\begin{itemize}


\item Model 1: $\S^{(1)}=I$.

\item Model 2: $\S^{(2)}=(\sigma^{(2)}_{ij})$ where $\sigma^{(2)}_{ij}=0.5^{|i-j|}$ for $1\leq i,j\leq p$. 

\item Model 3: $\S^{*(3)}=(\sigma^{*(3)}_{ij})$ where $\sigma^{*(3)}_{ii}=1$, $\sigma^{*(3)}_{ij}= \text{Unif}(1)*\text{Bernoulli}(1,0.02)$ for $i < j$ and $\sigma^{*(3)}_{ji}=\sigma^{*(3)}_{ij}$. $\S^{(3)}=(\S^{*(3)}+\delta I)/(1+\delta)$ with $\delta=\max\{-\lambda_{\min}(\S^{*(3)}),0\}+0.05$ to ensure positive definiteness.
\end{itemize}

The sample sizes are taken to be $n=100$ and 150, while the dimension $d$ varies over the values 20, 100 and 300. For each model, data are generated from multivariate distribution with mean zero and covariance matrix $\S$. Under the null hypothesis, the distribution is set to be multivariate normal, while under the alternative hypothesis, the distribution is set to be one of the following distributions.
\begin{itemize}

\item Distribution 1: Multivariate $t$ distribution with degrees of freedom $\nu = d/2$.


\item Distribution 2: Mixture Gaussian distribution $0.5 \mathcal{N}_d(\mathbf{0}, (1-a)\Sigma) + 0.5 \mathcal{N}_d(\mathbf{0}, (1+a)\Sigma)$ with $a=\frac{1.8}{\sqrt{d}}$.

\end{itemize}

We set the size of the tests to be 0.05 under all settings, and choose $B=500$ in the algorithm.  We run 1,000 replications to summarize the empirical size and power.  The empirical size results are reported in Table \ref{tab:size} and the power results of Distributions 1 and 2 are reported in Tables \ref{tab:power1} and \ref{tab:power2}. 

\begin{table}[htbp]
\caption{Empirical size (in percents) of the proposed algorithm (NEW), extended Friedman-Rafsky test (eFR), multivariate Shapiro--Wilk's test (mvSW) and the Fisher's test (Fisher). \label{tab:size}}
   \begin{center}
     \begin{tabular}{cccccccc}

     \toprule
     \multicolumn{2}{c}{$n$}   & \multicolumn{3}{c}{$100$}& \multicolumn{3}{c}{$150$}
              \\  \cmidrule(r){3-5} \cmidrule(r){6-8}

       \multicolumn{2}{c}{$d$} & \multicolumn{1}{c}{20}&
      \multicolumn{1}{c}{100}& \multicolumn{1}{c}{300}&
      \multicolumn{1}{c}{20}&
      \multicolumn{1}{c}{100}& \multicolumn{1}{c}{300}
       \\ \midrule

     \multirow{4}{*}{Model 1} 
     &NEW        & 4.3 & 4.3 & 5.1 & 3.9  & 4.8 & 5.6    \\
     & eFR         & 4.3 & 4.8 & 4.3 & 4.8  & 5.4 & 3.7    \\  
     & mvSW     & 6.4 & 4.6 & 5.1 & 4.0  & 4.8 & 5.3    \\  
     & Fisher      & 6.0 & 4.3 & 4.4 & 3.7  & 4.7 & 5.1    \\    \midrule

      \multirow{4}{*}{Model 2} 
      &NEW        & 5.9 & 4.3 & 6.4 & 5.9  & 5.8 & 4.9    \\
     & eFR         & 4.3 & 5.1 & 4.8 &  3.9 &7.0  & 6.9    \\  
     & mvSW     & 5.6 & 6.2 & 10.7 &  5.4 & 6.9 & 6.2    \\  
     & Fisher      & 5.3 & 6.6 & 10.4 & 4.4  & 7.0 & 5.6    \\    \midrule

      \multirow{4}{*}{Model 3} 
     &NEW        & 5.9 & 5.2 & 6.3 & 4.8  & 5.3 & 4.9    \\
     & eFR         & 5.4 & 4.8 & 19.7 & 4.8  & 4.0 & 15.2    \\  
     & mvSW     & 4.7 & 5.8 & 5.9 & 3.4  & 4.3 & 7.3    \\  
     & Fisher      & 5.2 & 5.3 & 5.3 & 3.8  & 4.4 & 6.1    \\    
        \bottomrule
    \end{tabular}

\end{center}

\end{table}

From Table \ref{tab:size}, we observe that the new test can control the size reasonably well under all settings, while the extended Friedman-Rafsky test has some serious size distortion for Model 3 when the dimension is larger than the sample size. { In addition, both of the multivariate Shapiro--Wilk's test and the Fisher's test have some size inflation for Model 2 when $d=300$ and $n=100$.}  
 
\begin{table}[htbp]
\caption{Empirical power (in percents) of the proposed algorithm (NEW), extended Friedman-Rafsky test (eFR), multivariate Shapiro--Wilk's test (mvSW) and the Fisher's test (Fisher) for multivariate $t$ distribution. \label{tab:power1}}
   \begin{center}
Multivariate $t$-distribution

\vspace{1em}  
  
     \begin{tabular}{cccccccc}

     \toprule
     \multicolumn{2}{c}{$n$}   & \multicolumn{3}{c}{$100$}& \multicolumn{3}{c}{$150$}
              \\  \cmidrule(r){3-5} \cmidrule(r){6-8}

       \multicolumn{2}{c}{$d$} & \multicolumn{1}{c}{20}&
      \multicolumn{1}{c}{100}& \multicolumn{1}{c}{300}&
      \multicolumn{1}{c}{20}&
      \multicolumn{1}{c}{100}& \multicolumn{1}{c}{300}
       \\ \midrule

     \multirow{4}{*}{Model 1} 
     &NEW        & 58.5 & 91.3 & 93.0 & 79.9  & 98.5 & 99.2    \\
     & eFR         & 6.7 & 3.7 & 6.4 & 10.2  & 3.9 & 4.7    \\  
     & mvSW     & 86.3 & 21.2 & 9.0 &  96.5 & 29.0 & 12.3    \\  
     & Fisher      & 86.2 & 20.8 & 8.7 & 96.7  & 30.1 &  12.1   \\    \midrule

      \multirow{4}{*}{Model 2} 
      &NEW        & 20.2 & 71.3 & 86.0 & 32.2  & 86.4 & 94.2    \\
     & eFR         & 11.8 & 5.4 & 5.2 & 15.0  & 4.6 & 6.0    \\  
     & mvSW     & 75.4 & 26.6 & 21.3 & 92.3  & 30.9 & 17.1    \\  
     & Fisher      & 75.8 & 27.2 & 20.9 & 92.6  & 31.2 & 16.0    \\    \midrule

      \multirow{4}{*}{Model 3} 
     &NEW        & 56.5 & 87.9 & 94.9 & 74.1  & 97.9 & 98.2    \\
     & eFR         & 6.7 & 4.8 & 18.4 & 11.0  & 3.5 & 11.2    \\  
     & mvSW     & 84.9 & 28.8 & 10.6 &  96.6 & 30.3 & 16.8    \\  
     & Fisher      & 85.6 & 28.5 & 10.8 & 96.5  & 31.0 & 17.3    \\    
      \bottomrule
    \end{tabular}
    
\end{center}

\end{table}

\begin{table}[htbp]
\caption{Empirical power (in percents) of the proposed algorithm (NEW), extended Friedman-Rafsky test (eFR), multivariate Shapiro--Wilk's test (mvSW) and the Fisher's test (Fisher) for mixture Gaussian distribution. \label{tab:power2}}
   \begin{center}

 Mixture Gaussian distribution
 
 \vspace{1em}
 
     \begin{tabular}{cccccccc}

     \toprule
     \multicolumn{2}{c}{$n$}   & \multicolumn{3}{c}{$100$}& \multicolumn{3}{c}{$150$}
              \\  \cmidrule(r){3-5} \cmidrule(r){6-8}

       \multicolumn{2}{c}{$d$} & \multicolumn{1}{c}{20}&
      \multicolumn{1}{c}{100}& \multicolumn{1}{c}{300}&
      \multicolumn{1}{c}{20}&
      \multicolumn{1}{c}{100}& \multicolumn{1}{c}{300}
       \\ \midrule

     \multirow{4}{*}{Model 1} 
     &NEW        & 45.8 & 81.7 & 81.6 & 66.3  & 95.8 & 94.8    \\
     & eFR         & 6.6 & 3.9 & 5.3 & 7.4  & 3.7 & 5.7    \\  
     & mvSW     & 55.9 & 12.5 & 8.0 & 68.7  & 17.1 & 10.5    \\  
     & Fisher      & 56.0 & 12.7 & 8.0 & 69.8  & 16.8 & 10.3    \\    \midrule

      \multirow{4}{*}{Model 2} 
      &NEW        & 15.3 & 61.9 & 71.9 & 26.7  & 67.2 & 81.9    \\
     & eFR         & 6.3 & 5.6 & 4.9 & 10.4  & 6.4 & 6.5    \\  
     & mvSW     & 46.3 & 17.7 & 19.2 & 63.3  & 19.7 & 12.7    \\  
     & Fisher      & 47.3 & 18.3 & 19.2 & 63.8  & 19.8 & 12.3    \\    \midrule

      \multirow{4}{*}{Model 3} 
     &NEW        & 45.0 & 75.4 & 86.9 & 64.0  & 90.8 & 94.7    \\
     & eFR         & 6.5 & 4.5 & 21.7 & 7.5  & 3.8 & 14.9    \\  
     & mvSW     & 53.2 & 19.0 & 10.5 & 70.1  & 18.9 & 13.9    \\  
     & Fisher      & 55.0 & 19.3 & 9.8 & 70.5  & 18.5 & 13.2    \\    
      \bottomrule
    \end{tabular}

\end{center}

\end{table}

For power comparison, we first studied the annoying heavy tail scenario -- multivariate $t$-distribution.  It can be seen from Table \ref{tab:power1} that, the new test can capture the signal very well, while the extended Friedman-Rafsky test suffers from much lower power. { In the meanwhile, both of the multivariate Shapiro--Wilk's test and the Fisher's test have competitive power performance under the low-dimensional settings, but have fast decaying power performance (much lower than the proposed method) as $d$ increases.}  We also studied the scenario that the distribution is a mixture of two multivariate Gaussian distributions and we observed similar phenomena in Table \ref{tab:power2} that the new test has much higher power than the extended Friedman-Rafsky test under all settings and has better performance than ``mvSW'' and ``Fisher'' for $d=100$ and $300$. 

{ The empirical size and power performance of all four methods are also illustrated in the empirical cumulative distribution function (ecdf) plots as shown in Figures \ref{ecdf1.fig} and \ref{ecdf2.fig}, for Model 1 and $d=n=100$. We observe similar patterns for the other models.} In summary, for all scenarios studied above, our newly proposed algorithm provides superior performance in both empirical size as well as empirical power comparing with the existing methods.

\begin{figure}[h]
 \caption{Empirical size cdf plots of the four methods for Model 1, $d=n=100$. }
 \includegraphics[width=1\textwidth]{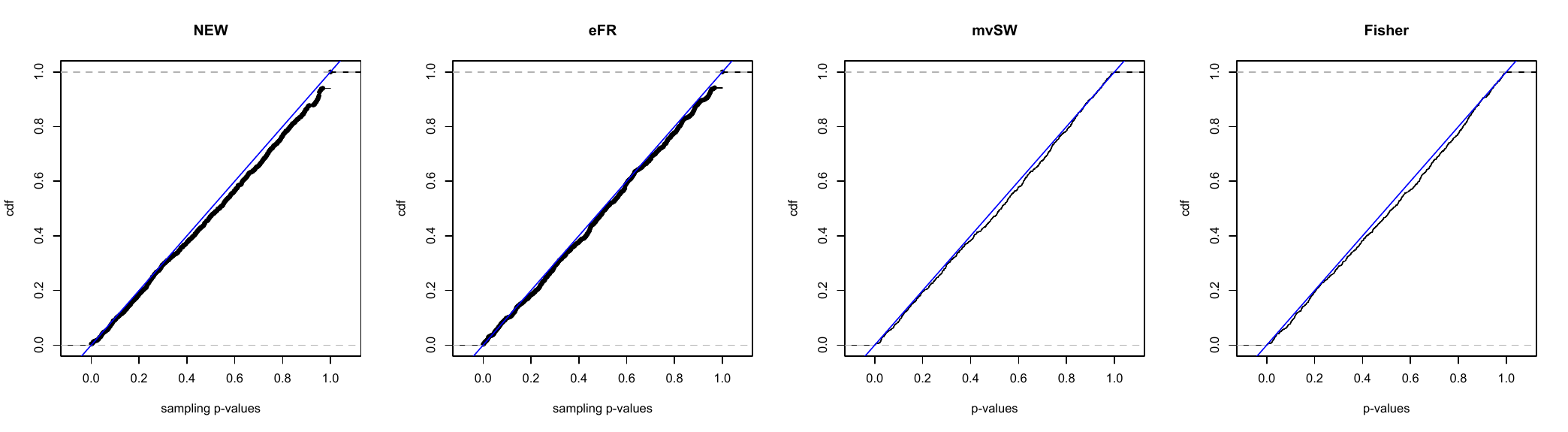} 
 \label{ecdf1.fig}
\end{figure}

\begin{figure}[htbp]
 \caption{Empirical power cdf plots of the four methods for Model 1, Distribution 1, $d=n=100$. }
 \includegraphics[width=1\textwidth]{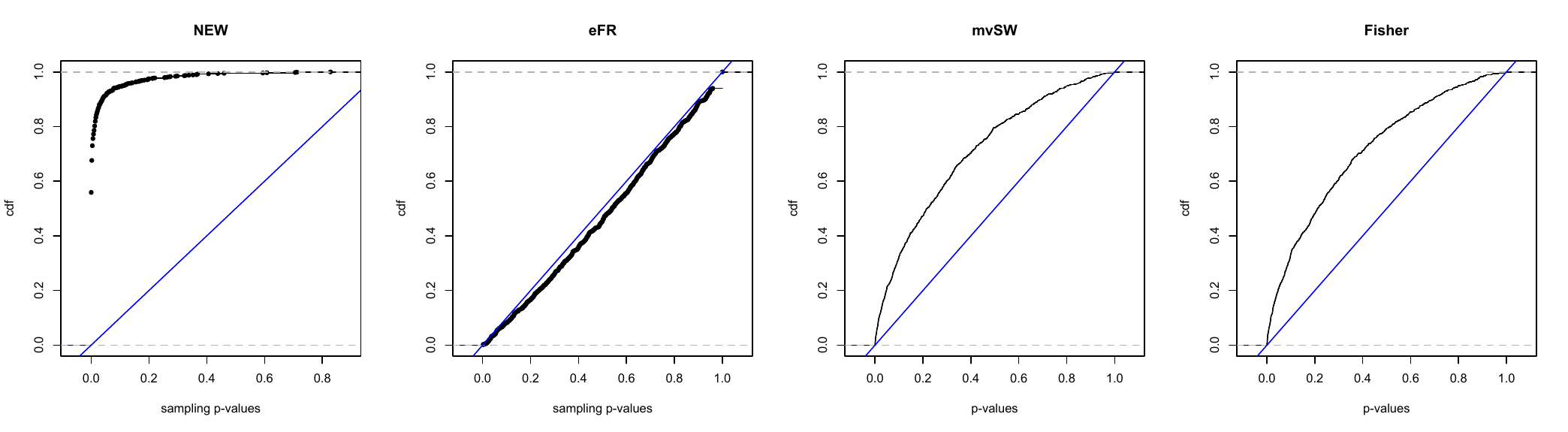} 
 \label{ecdf2.fig}
\end{figure}
\section{Application}
\label{sec:application}

Classification is an important statistical problem that has been extensively studied both in the traditional low-dimensional setting and the recently developed high-dimensional setting. In particular, Fisher's linear discriminant analysis has been shown to perform well and enjoy certain optimality as the sample size tends to infinity while the dimension is fixed \citep{anderson1954introduction}, and it has also been widely studied in the high-dimensional setting when the sample covariance matrix is no longer invertible, see, e.g., \cite{bickel2004some}, \cite{fan2008high}, \cite{cai2011direct} and \cite{mai2012direct}. In all of those studies, normality of the data is a key assumption in order to obtain the linear discriminant rule and investigate the subsequent analysis of misclassification rate. { We study in this section a lung cancer data set, which was analyzed by \cite{gordon2002translation} and is available at \texttt{R} documentation \texttt{data(lung)} { with package \texttt{propOverlap}}. This data set was popularly used in the classification literature \citep{fan2008high,cai2011direct} where normality is a key assumption. In addition, we explore a data set that was analyzed in \cite{jin2016influential} by their method IF-PCA for clustering, where data normality is assumed.}


\subsection{Lung cancer data}
The lung cancer data set has 181 tissue samples, including 31 malignant pleural mesothelioma (MPM) and 150 adenocarcinoma (ADCA), and each sample is described by 12533 genes. This data set has been analyzed in \cite{fan2008high} by their methods FAIR and NSC, and in \cite{cai2011direct} by their LPD rule, for distinguishing MPM from ADCA, which is important and challenging from both clinical and pathological perspectives. However, before applying their proposed methods, none of them have checked the normality of the data, which is a fundamental assumption in the formulation of linear discriminants. If the normality fails to hold, then the misclassification rates can be affected and their results may no longer be valid.

In this section, we use our newly developed method to check the normality of the 150 ADCA samples in this lung cancer data set. 
Note that, multivariate normality assumption for the 12533 genes of the ADCA samples will be rejected if any subset for this large number of genes deviate from the normality.  Thus, we randomly select a group of 200 genes, and applied our new method to test the multivariate normality assumption. By applying Algorithm \ref{algo1} with $B=500$, we obtain that, the sampling $p$-value is equal to 0, which gives sufficient evidence that the samples from this data set have severe deviation from the multivariate normal distribution.  We further repeat this procedure for 100 times.  In each time, we randomly select a group of 200 genes and apply Algorithm 1 ($B=500$) to the selected genes.  It turns out that the sampling $p$-values are all 0 for these 100 times.   Thus, it is not reasonable to assume the normality and directly apply the recent developed high-dimensional linear discriminant procedures to classify MPM and ADCA, as studied in \cite{fan2008high} and \cite{cai2011direct}. So our procedure serves as an important initial step for checking the normality assumption before applying any statistical analysis methods which assume such conditions.

{
\subsection{Colon cancer data}
Next, we study in this section a gene expression data set on tumor and normal colon tissues that was analyzed and cleaned by \cite{dettling2004bagboosting}.  This data set can be found at \url{https://blog.nus.edu.sg/staww/softwarecode/}.  It has 40 tumor and 22 normal colon tissue samples, and each sample is described by 2000 genes.  This data set has been analyzed in \cite{jin2016influential} by their method IF-PCA for clustering, where they imposed normality assumption on the data, though they found the violation to such assumption in their analysis as the empirical null distribution of a test statistic they used was far from the theoretical null distribution derived from the normal assumption. 

In this section, we use our newly developed method to check the normality of the 40 tumor samples in this colon cancer data set. 
We compare the proposed method with eFR, the multivariate Shapiro--Wilk's test and the Fisher’s test in this analysis.
By applying Algorithm \ref{algo1} with $B=500$, we obtain that, the sampling $p$-value is equal to 0, which gives a sufficient evidence that the samples from this data set have severe deviation from the multivariate normal distribution.  This double confirms the deviation from the normality assumption noticed by the authors in \cite{jin2016influential}.   On the other hand, both the multivariate Shapiro--Wilk's test and the Fisher’s test successfully reject the null while the eFR method reports a sampling $p$-value of 1 and fails to detect the violation to the normality assumption.

}

\section{Discussion}
\label{sec:discussion}
We proposed in this paper a nonparametric normality test based on the nearest neighbor information. It enjoys proper error control and is shown to have significant power improvement over the alternative approaches. We discuss in this section a related test statistic and some extensions and explorations of the current method.

\subsection{Test statistic based on $XX$}

Our proposed test statistic involves the event $YY$, i.e., the event that an observation in $\{Y_1,Y_2,\dots,Y_n\}$ finds its nearest neighbor in $\{Y_1,Y_2,\dots,Y_n\}$. A straightforward alternative method could be based on the test statistics which involves the event $XX$, i.e., the event that an observation in $\{X_1,X_2,\dots,X_n\}$ finds its nearest neighbor in $\{X_1,X_2,\dots,X_n\}$, and a question is whether the $XX$-equivalent statistic could be incorporated to further enhance the power.  Unfortunately, the $XX$ version is not as robust as the $YY$ version and does not have good performance in controlling the type I error.  Table \ref{tab:sizeXX} lists the empirical size of the $XX$ version of the test under the same settings as in Table \ref{tab:size}.  We observe that this statistic has serious size distortion  for Model 3 when the dimension is high.  This also explains the bad performance of eFR in controlling type I error under Model 3 because eFR partially uses the $XX$ information.

\begin{table}[htbp]
\caption{Empirical size (in percents) of the $XX$ version test, $\alpha=0.05$. \label{tab:sizeXX}}
   \begin{center}
     \begin{tabular}{ccccccc}

     \toprule
       \multicolumn{1}{c}{$n$}   & \multicolumn{3}{c}{$100$}& \multicolumn{3}{c}{$150$}
              \\  \cmidrule(r){2-4} \cmidrule(r){5-7}

       \multicolumn{1}{c}{$d$} & \multicolumn{1}{c}{20}&
      \multicolumn{1}{c}{100}& \multicolumn{1}{c}{300}&
      \multicolumn{1}{c}{20}&
      \multicolumn{1}{c}{100}& \multicolumn{1}{c}{300}
       \\ \midrule

  	Model 1 & 4.6 & 3.5 & 5.1 & 4.3 & 3.3 & 4.5 \\ \midrule
	Model 2 & 5.3 & 5.7 & 8.6 & 5.5 & 4.0 & 9.2 \\ \midrule
	Model 3 & 4.4 & 5.1 & 32.2 & 4.0 & 4.5 & 16.2 \\ \bottomrule
        \end{tabular}

\end{center}

\end{table}

\subsection{Extension to other distributions in the exponential family}

The idea of constructing this normality test could be extended to other distributions in the exponential family.  As long as one has reasonably good estimators for the parameters of the distribution, a similar procedure as described in Section \ref{sec:meth} can be applied.  In particular, one could replace the multivariate normal distribution in Algorithm \ref{algo1} by the distribution of interest, and replace the mean and covariance estimators by the estimators of the corresponding parameters.  The conditions for the asymptotic equivalence between the events $YY$ and $Y^*Y^*$ would need more careful investigations and warrant future research.




{
\subsection{A power enhanced algorithm}
\label{sec:L10}

To further improve the power performance of the method, especially when the sample size is limited, we can increase the number of sampling procedure in Algorithm \ref{algo1} as detailed in the following algorithm.

\begin{algorithm}[H]
\caption{}
\label{algo2}
\begin{algorithmic}[1]
\item For $i=1:L$, generate $Y_{i,1}, \dots, Y_{i,n} \overset{iid}{\sim} \mathcal{N}_d(\mu_x, \Sigma_x)$, calculate $r_i(YY)$. Let $\bar r(YY)$ be the average of $r_i(YY)$'s.

\item \label{step2} Generate $X_1^*, \dots, X_n^* \overset{iid}{\sim} \mathcal{N}_d(\mu_x, \Sigma_x)$, estimate its mean $\mu_{x^*}$ and covariance matrix $\Sigma_{x^*}$.  
For $i=1:L$, generate $Y_{i,1}^*, \dots, Y_{i,n}^* \overset{iid}{\sim} \mathcal{N}_d(\mu_{x^*}, \Sigma_{x^*})$, calculate $r_i(Y^*Y^*)$. Let $\bar r(Y^*Y^*)$ be the average of $r_i(Y^*Y^*)$'s.

\item Repeat Step 2 for $B$ times to get an estimate of the empirical distribution of $\bar r(YY)$ under $H_0$.

\item Compute the two-sided {\it sampling $p$-value}, $p(YY)$, i.e., the percentage of $|\bar r(Y^*Y^*)- m(\bar r(Y^*Y^*))|$ (out of $B$) that are larger than or equal to $|\bar r(YY)- m(\bar r(Y^*Y^*))|$, where $m(\bar r(Y^*Y^*))$ is the average of $\bar r(Y^*Y^*)$'s in Step 3.

\item For a given significance level $0<\alpha<1$, define $\Psi_\alpha = I\{p(YY) \leq \alpha\}$. Reject the null hypothesis whenever $\Psi_\alpha = 1$.
\end{algorithmic}
\end{algorithm}

Note that, when $L=1$, Algorithm \ref{algo2} is reduced to Algorithm \ref{algo1} in Section \ref{sec:algo}. In the following Figure \ref{boxplot.fig}, we show the boxplots of the sampling $p$-values for Distribution 1 and Model 3 when $L=1,2,\ldots, 10$, for $d=n=100$, with 100 replications. Similar patterns are observed for the other models. It can be seen that, as $L$ increases, the power performance of the method can be significantly improved and will get stable when $L$ is around 5. Also note that, the computation cost is growing as $L$ increases. Hence we mainly recommend Algorithm \ref{algo1} in the paper as it already shows reasonable well performance both in terms of empirical size and power as illustrated in Section \ref{sec:simulation}. 

\begin{figure}[h]
 \caption{\small Boxplots of the sampling $p$-values for Distribution 1 and Model 3 in Section \ref{sec:simulation}, for $d=m=100$, with 100 replications.}
\includegraphics[width=0.7\textwidth]{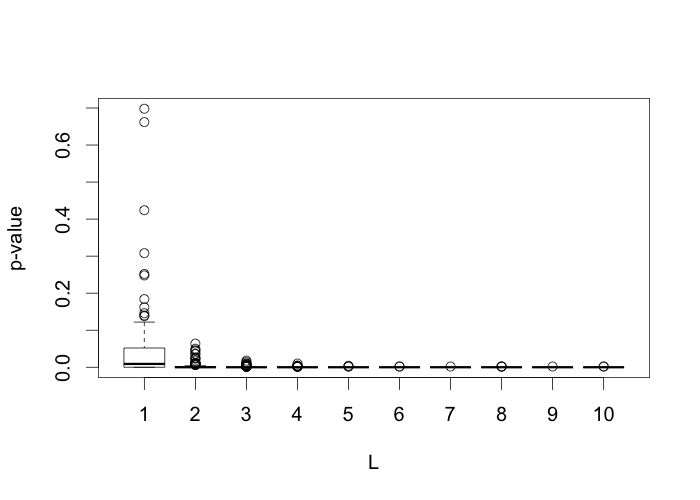} 
\centering
 \label{boxplot.fig}
\end{figure}

\subsection{Non-significant results and scale sensitivity}
When the testing results are nonsignificant, it means that the distribution is very close to the multivariate normal, whereas the unbiased property of the proposed test warrants future research. In this section, we perform additional simulation studies to explore the power performance of the proposed method as the distributions are approaching to normality. Specifically, we consider the following three sets of distributions:  multivariate chi-square distributions, multivariate $t$ distributions and multivariate Gaussian distribution with a certain proportion of the dimensions replaced by $t$ distribution.
\begin{itemize}

\item Multivariate chi-square distribution with degrees of freedom $\nu = 3, 5, 10$ and $20$ (the larger $\nu$ is, the closer the distribution is to the multivariate normal distribution).

\item Multivariate $t$ distribution with degrees of freedom $\nu = d/4, d/2, d$ and $2d$ (the larger $\nu$ is,  the closer the distribution is to the multivariate normal distribution).

\item Multivariate Gaussian distribution $\mathcal{N}_d(\mathbf{0}, \Sigma)$ with a certain proportion of the dimensions, ranging from $0.5$ to $0.1$, replaced by Multivariate $t$ distribution with degrees of freedom $\nu = d/4$ (the smaller the proportion is,  the closer the distribution is to the multivariate normal distribution).

\end{itemize}
The power performance of the methods are summarized in Tables \ref{tab_dof2} and \ref{tab:prop}.  It can been seen from the tables that, when the alternatives are getting closer to the multivariate normal, the testing results become more and more non-significant.

\begin{table}[htbp]
\caption{Empirical power (in percents) of the proposed algorithm (NEW) and extended Friedman-Rafsky test (eFR) for Model 3, $\alpha=0.05$, $d=100$ and $n=100$.}
   \begin{center}
     \begin{tabular}{ccccccccc}
     \toprule
     \multicolumn{1}{c}{Distribution}   & \multicolumn{4}{c}{chi-square}& \multicolumn{4}{c}{t}
              \\  \cmidrule(r){2-5} \cmidrule(r){6-9}

       \multicolumn{1}{c}{DOF} & \multicolumn{1}{c}{3}&
      \multicolumn{1}{c}{5}& \multicolumn{1}{c}{10}&
      \multicolumn{1}{c}{20}& \multicolumn{1}{c}{$d/4$}&
      \multicolumn{1}{c}{$d/2$}& \multicolumn{1}{c}{$2d$}
      & \multicolumn{1}{c}{$4d$}
       \\ \midrule

    NEW      & 45.1 & 25.2 & 10.6 & 6.5 & 99.7 & 87.9 & 17.3 & 6.3    \\
      eFR         & 9.4 & 6.8 & 6.5 & 5.1 & 2.3 & 4.8 & 4.0 & 3.3   \\  
        \bottomrule
    \end{tabular}

\end{center}
\label{tab_dof2}
\end{table}

\begin{table}[htbp]
\caption{Empirical power (in percents) with varying proportion of non-Gaussian dimensions, ranging from $0.5$ to $0.1$, for Model 1, $\alpha=0.05$, $d=100$ and $n=100$. \label{tab:prop}}
   \begin{center}
     \begin{tabular}{cccccc}

     \toprule
        \multicolumn{1}{c}{non-Gaussian proportion} & \multicolumn{1}{c}{0.5}&
      \multicolumn{1}{c}{0.4}& \multicolumn{1}{c}{0.3}&
      \multicolumn{1}{c}{0.2}&
      \multicolumn{1}{c}{0.1}
       \\ \midrule

     NEW         & 48.1 & 25.9  & 12.9 & 7.0 & 4.2    \\
     eFR           & 3.7 & 5.2  & 4.8  & 5.0  & 4.1   \\




        \bottomrule
    \end{tabular}

\end{center}

\end{table}

In addition, we explore the scale sensitivity of the proposed test by considering different covariance models with varying condition numbers as follows.
\begin{itemize}


\item Model 1A: $\S=I$.

\item Model 1B: $\S=\text{diag}(\sigma_{1,1},\ldots, \sigma_{d,d})$, where $\sigma_{i,i}= \text{Unif}(1, 5)$ for $i=1,\ldots,d$. 

\item Model 1C: $\S=\text{diag}(\sigma_{1,1},\ldots, \sigma_{d,d})$, where $\sigma_{i,i}= \text{Unif}(1, 20)$ for $i=1,\ldots,d$. 
\end{itemize}
We see from Table \ref{tab_condition} that, the empirical size and power performance (Distribution 1) of these three models are very similar to each other, which shows that the proposed method is not sensitive to the scale of the data.

\begin{table}[htbp]
\caption{Empirical size and power (in percents) of the proposed algorithm (NEW) and extended Friedman-Rafsky test (eFR) for varying condition numbers, $n=100$. }
   \begin{center}
     \begin{tabular}{cccccccc}

     \toprule
     \multicolumn{2}{c}{}   & \multicolumn{3}{c}{Size}& \multicolumn{3}{c}{Power}
              \\  \cmidrule(r){3-5} \cmidrule(r){6-8}

       \multicolumn{2}{c}{$d$} & \multicolumn{1}{c}{20}&
      \multicolumn{1}{c}{100}& \multicolumn{1}{c}{300}&
      \multicolumn{1}{c}{20}&
      \multicolumn{1}{c}{100}& \multicolumn{1}{c}{300}
       \\ \midrule

     \multirow{2}{*}{Model 1A} &NEW      & 4.3 & 4.3 & 5.1 & 58.5  & 91.3 & 93.0    \\
      & eFR         & 4.3 & 4.8 & 4.3 & 6.7  & 3.7 & 6.4    \\  \midrule

      \multirow{2}{*}{Model 1B} &NEW      & 3.2 & 4.7 & 5.0 & 55.9  & 91.2 & 91.7    \\
     & eFR         & 4.0 & 3.6 & 5.6 & 7.6  & 3.9 &  4.4   \\  \midrule

      \multirow{2}{*}{Model 1C} &NEW      & 5.1 & 5.1 & 5.3 & 49.8  & 87.0 & 92.0    \\
     & eFR         & 4.7 & 5.3 & 4.7 & 7.1  & 3.6 &  4.7   \\  
        \bottomrule
    \end{tabular}

\end{center}
\label{tab_condition}
\end{table}
}

\section{Proof of Theorem \ref{thm:asym}}
\label{sec:thmproof}
Let $\S=U\Lambda U^T$ and $\S_x=U_x\Lambda_xU_x^T$ be respectively the eigen-decomposition  of $\S$ and $\S_x$. Define $\S^{1/2}=U\Lambda^{1/2}U^T$ and $\S_x^{1/2}=U_x\Lambda_x^{1/2}U_x^T$. Then under the conditions of Theorem \ref{thm:asym}, by Lemma \ref{lemma:square_root}, we have $\|\Sigma_x^{1/2}-\Sigma^{1/2}\|_2 = o_{\pr}(n^{-\frac{1}{d}-\frac{(2+a) \log d + \kappa}{2 \log n}})$.

Let $f(\cdot)$ be the density of $\mathcal{N}_d(\mu,\Sigma)$, and $f^*(\cdot)$ be the density of $\mathcal{N}_d(\mu_x, \Sigma_x)$. 
Then we have,
\begin{align*}
& \pr(YY)  =  \int \pr(YY|\{X_i=x_i\}_{i=1,\dots,n}) \prod_{i=1}^n f(x_i) dx_i, \\
& \pr(Y^*Y^*)  =  \int \pr(Y^*Y^*|\{X^*_i=x^*_i\}_{i=1,\dots,n}) \prod_{i=1}^n f^*(x^*_i) dx^*_i.
\end{align*}
By the construction of $\{Y_1,\dots, Y_n\}$ and $\{Y^*_1,\dots, Y^*_n\}$, we have 
$$ \pr(YY|\{X_i=x_i^*\}_{i=1,\dots,n}) =  \pr(Y^*Y^*|\{X^*_i=x_i^*\}_{i=1,\dots,n}).$$
Hence,
$$\pr(Y^*Y^*)  =  \int \pr(YY|\{X_i=x_i^*\}_{i=1,\dots,n}) \prod_{i=1}^n f^*(x_i^*) dx_i^*.$$
By a change of measure, we have
$$\pr(Y^*Y^*) =  \int \pr(YY|\{X_i=\Sigma_x^{1/2}\Sigma^{-1/2}(x_i-\mu)+\mu_x\}_{i=1,\dots,n}) \prod_{i=1}^n f(x_i) dx_i. $$
It is not hard to see that if we shift the $x_i$'s all by a fixed value, the probability of $YY$ is unchanged.  Hence,
$$\pr(Y^*Y^*) =  \int \pr(YY|\{X_i=\Sigma_x^{1/2}\Sigma^{-1/2}x_i\}_{i=1,\dots,n}) \prod_{i=1}^n f(x_i) dx_i. $$
Let $w_i = \Sigma_x^{1/2}\Sigma^{-1/2}x_i$.  Then,
\begin{align}
& |\pr(YY)-\pr(Y^*Y^*)|  \nonumber \\
& = \left|\int \left(\pr(YY|\{X_i=x_i\}_{i=1,\dots,n})-\pr(YY|\{X_i=w_i\}_{i=1,\dots,n})\right) \prod_{i=1}^n f(x_i) dx_i \right| \label{eq:diff} \\
& \leq \int \left|\pr(YY|\{X_i=x_i\}_{i=1,\dots,n})-\pr(YY|\{X_i=w_i\}_{i=1,\dots,n})\right| \prod_{i=1}^n f(x_i) dx_i. \label{eq:diff2}
\end{align}

Let $\mu_1$ and $\Sigma_1$ be the estimated mean and variance based on $\{x_i\}_{i=1,\dots,n}$, and $\mu_2$ and $\Sigma_2$ be the estimated mean and variance based on $\{w_i\}_{i=1,\dots,n}$.   
Let $g_1(\cdot)$ and $g_2(\cdot)$ be the density function of $\mathcal{N}_d(\mu_1,\Sigma_1)$ and $\mathcal{N}_d(\mu_2,\Sigma_2)$, respectively. Let $Y_1, Y_2, \dots, Y_n \overset{iid}{\sim} \mathcal{N}_d(\mu_1,\Sigma_1)$, $\tilde{Y}_1, \tilde{Y}_2, \dots, \tilde{Y}_n \overset{iid}{\sim} \mathcal{N}_d(\mu_2,\Sigma_2)$, $N_{Y_i}$ be the observation in $\{\{Y_j\}_{j\neq i}, \{x_j\}_{j=1,\dots,n}\}$ that is closest to $Y_i$, $N_{\tilde{Y}_i}$ be the observation in $\{\{\tilde{Y}_j\}_{j\neq i}, \{w_j\}_{j=1,\dots,n}\}$ that is closest to $\tilde{Y}_i$, $\{Y\}=\{Y_i\}_{i=1,\dots,n}$, and $\{\tilde{Y}\}=\{\tilde{Y}_i\}_{i=1,\dots,n}$.  Then,  
\begin{align*}
& \pr(YY|\{X_i=x_i\}_{i=1,\dots,n})-\pr(YY|\{X_i=w_i\}_{i=1,\dots,n}) \\
& = \pr(N_{Y_1}\in\{Y\}) - \pr(N_{\tilde{Y}_1}\in \{\tilde{Y}\}) \\
& = \int \pr(N_{Y_1}\in\{Y\}|Y_1=y)g_1(y)dy - \int \pr(N_{\tilde{Y}_1}\in \{\tilde{Y}\}|\tilde{Y}_1=\tilde{y})g_2(\tilde{y}) d\tilde{y}.	 
\end{align*}
By change of measure, we have that
\begin{align*}
	& \int \pr(N_{\tilde{Y}_1}\in \{\tilde{Y}\}|\tilde{Y}_1=\tilde{y})g_2(\tilde{y}) d\tilde{y}\\
	& = \int \pr(N_{\tilde{Y}_1}\in \{\tilde{Y}\}|\tilde{Y}_1=\Sigma_2^{1/2}\Sigma_1^{-1/2}(y-\mu_1)+\mu_2)g_1(y) dy. 
\end{align*}
 
Let $y_w = \Sigma_2^{1/2}\Sigma_1^{-1/2}(y-\mu_1)+\mu_2$, then
\begin{align}
& |\pr(YY|\{X_i=x_i\}_{i=1,\dots,n})-\pr(YY|\{X_i=w_i\}_{i=1,\dots,n})|   \nonumber \\
& = \left|\int (\pr(N_{Y_1}\in\{Y\}|Y_1=y) - \pr(N_{\tilde{Y}_1}\in \{\tilde{Y}\}|\tilde{Y}_1=y_w)) g_1(y) dy \right| 
\label{eq:diffy} \\
& \leq 	\int |\pr(N_{Y_1}\in\{Y\}|Y_1=y) - \pr(N_{\tilde{Y}_1}\in \{\tilde{Y}\}|\tilde{Y}_1=y_w))| g_1(y) dy.  \label{eq:diffy2}
\end{align}


Let  $r_i = \|w_i - x_i\|_2$.  Define $\alpha^* = -\frac{1}{d}-\frac{(1+a)\log d+\kappa}{2\log n}$. By Lemma \ref{lemma:xdist} (see Section \ref{sec:lemma}), we have that $r_i = o_{\pr}(n^{\alpha^*}), \forall i$.  Also, given that $w_i = \Sigma_x^{1/2} \Sigma^{-1/2} x_i$, it is easy to have estimates such that $\mu_2 = \Sigma_x^{1/2} \Sigma^{-1/2}\mu_1$. 
 Then we have
\begin{eqnarray}
\|y_w-y\|_2 &=& \|\Sigma_2^{1/2}\Sigma_1^{-1/2}(y-\mu_1)+\mu_2-y\|_2\cr
&=& \|(\Sigma_2^{1/2}-\Sigma_1^{1/2})\Sigma_1^{-1/2}y + (\Sigma_x^{1/2}\Sigma^{-1/2}-\Sigma_2^{1/2}\Sigma_1^{-1/2})\mu_1\|_2\cr
&\leq& \|(\Sigma_2^{1/2}-\Sigma_1^{1/2})\Sigma_1^{-1/2}y\|_2 + \|(\Sigma_x^{1/2}\Sigma^{-1/2}-\Sigma_2^{1/2}\Sigma_1^{-1/2})\mu_1\|_2
\end{eqnarray}

Denote by $\tilde{\alpha} = {-\frac{1}{d}-\frac{(2+a) \log d + \kappa}{2 \log n}}$. 
Recall that $\|\Sigma_x^{1/2}-\Sigma^{1/2}\|_2 = o_{\pr}(n^{\tilde{\alpha}})$. Note that, the covariance matrix of $\{x_i, i=1,\ldots,n\}$ is $\Sigma$ and the covariance matrix of $\{w_i, i=1,\ldots,n\}$ is $\Sigma_x$. Then using the same estimation method of the covariance matrix as estimating $\Sigma$ by $\Sigma_x$, we can estimate $\Sigma_x$ by an estimator $\Sigma_2$ and estimate $\Sigma$ by $\Sigma_1$, such that
\[
\|\Sigma_2-\Sigma_x\|_2 = o_\pr(n^{\tilde{\alpha}}), \mbox{ and }\|\Sigma_1-\Sigma\|_2 = o_\pr(n^{\tilde{\alpha}}).
\]
Note that $\|\Sigma_x-\Sigma\|_2 = o_\pr(n^{\tilde{\alpha}})$, we have that $\|\Sigma_2-\Sigma_1\|_2 = o_\pr(n^{\tilde{\alpha}})$.
Then by the proofs of Lemma \ref{lemma:square_root} and the conditions of Theorem \ref{thm:asym}, we have that $\|\Sigma_2^{1/2}-\Sigma_1^{1/2}\|_2 = o_\pr(n^{\tilde{\alpha}})$.

Thus we have 
\begin{align*}
\|(\Sigma_2^{1/2}-\Sigma_1^{1/2})\Sigma_1^{-1/2}y\|_2 & \leq \|\Sigma_2^{1/2}-\Sigma_1^{1/2}\|_2\|\Sigma_1^{-1/2}y\|_2 \\
& = o_\pr(n^{\tilde{\alpha}}) O_\pr(\sqrt{d}) = o_{\pr} (n^{\alpha^*}).	
\end{align*}

By similar arguments, we have
\begin{eqnarray}
&&\|(\Sigma_x^{1/2}\Sigma^{-1/2}-\Sigma_2^{1/2}\Sigma_1^{-1/2})\mu_1\|_2 \cr
&& = \|((\Sigma_x^{1/2}-\Sigma^{1/2})\Sigma^{-1/2}-(\Sigma_2^{1/2}-\Sigma_1^{1/2})\Sigma_1^{-1/2})\mu_1\|_2 \cr
&& \leq \|((\Sigma_x^{1/2}-\Sigma^{1/2})\Sigma^{-1/2}\mu_1\|_2+\|(\Sigma_2^{1/2}-\Sigma_1^{1/2})\Sigma_1^{-1/2})\mu_1\|_2\cr
&& = o_{\pr} (n^{\alpha^*}). 
\end{eqnarray}
Thus we have that $\|y_w-y\|_2 = o_{\pr}(n^{\alpha^*})$.
 
Let
\begin{align*}
	j_x & = \arg \min_{j \in \{1,2, \dots, n\}} \{ \|y-x_j\|_2\}, \\
	j_w & = \arg \min_{j \in \{1,2, \dots, n\}} \{ \|y_w-w_j\|_2\}. \end{align*}
and $D_{\min,x} = \|y-x_{j_x}\|_2$, $D_{\min,w} =  \|y_w-w_{j_w}\|_2$.  
 
Suppose $D_{\min,x} = O_{\pr}(n^\alpha)$. Notice that $n^{\alpha^*} = n^{-\frac{1}{d}} d^{-(1+a)/2}e^{-\kappa /2} \leq O(d^{-1/2})$.   When $\alpha<\alpha^*$, based on Lemma \ref{lemma:onept_ori}, the probability that $D_{\min,x} = c_0 n^\alpha$ for some constant $c_0>0$ is of order $n\times o_\pr(n^{\alpha^* d} d^{-d/2} e^{\kappa d/2}) = o_\pr(d^{-d}) = o_\pr(1)$.

We thus focus on $\alpha\geq \alpha^*$.  By definitions of $D_{\min,x}$ and $D_{\min,w}$, and the facts that $\|x_i-w_i\|_2=o_\pr(n^{\alpha^*}), \forall i$, and $\|y-y_w\|_2 = o_\pr(n^{\alpha^*})$, we have that 
 $D_{\min,w} = D_{\min,x} + o_\pr(n^{\alpha^*})$.  Let $p_x$ be the probability that $Y_k\sim \mathcal{N}_d(\mu_1,\Sigma_1)$ falls in the $D_{\min,x}$-ball of $y$, and $p_w$ be the probability that $\tilde{Y}_k\sim \mathcal{N}_d(\mu_2,\Sigma_2)$ falls in the $D_{\min,w}$-ball of $y_w$.

 Let $\alpha_0 = -\frac{1}{d}+\frac{(1-a)\log d - \kappa}{2\log n} > \alpha^*$. We consider two scenarios: (1) $\alpha^*\leq\alpha<\alpha_0$, and (2) $\alpha\geq \alpha_0$.
	\begin{enumerate}[(1)]
		\item $\alpha^*\leq\alpha<\alpha_0$:
		\begin{enumerate}
			\item When $d\log d\leq \log n$, we have $n^{\alpha_0} = d^{\frac{1}{2}-\frac{\log n}{d\log d} - \frac{\kappa}{2\log d}}\leq O(d^{-\frac{1}{2}})$. Since $\mu_1$ and $\Sigma_1$ satisfy the condition for Lemma \ref{lemma:onept}, we have $$p_x = o_\pr(n^{d\alpha_0} d^{-d/2} e^{\kappa_1 d/2}) = o_\pr(n^{-1}\sqrt{|\Sigma|/|\Sigma_1|}) = o_\pr(n^{-1}),$$ where $\kappa_1 = 1-\frac{\log |\Sigma_1|}{d}-\log 2$.
			
			 Notice that $\mu_2$ and $\Sigma_2$ also satisfy the condition for Lemma \ref{lemma:onept}, so $$p_w = o_\pr(n^{d\alpha_0} d^{-d/2} e^{\kappa_2 d/2})= o_\pr(n^{-1}\sqrt{|\Sigma|/|\Sigma_2|}) = o_\pr(n^{-1}),$$ where $\kappa_2 = 1-\frac{\log |\Sigma_2|}{d}-\log 2$. 
			 \item When $d\log d>\log n$ and $d=o(\log n)$, we have $n^{\alpha_0}\leq O(\sqrt{d})$, by Lemma \ref{lemma:onept}, $\log p_x$ is 
	 \begin{align*}
	 	 & -\tfrac{1}{2}d \log d + d \log D_{\min,x} + \tfrac{1}{2}d(\kappa_1 + O_\pr(1)) \\
	 	 & \hspace{5mm} = -\log n - \tfrac{d}{2}(a\log d + O_\pr(1)).
	 \end{align*}
     Here, $a=\frac{\log n}{\xi_{d,n} d \log d}$ with $1\ll\xi_{d,n}=o({\log n}/{d})$ a positive constant.  We have $\log p_x$ is $-\log n - \frac{1}{2\xi_{d,n}}\log n + O_\pr(d)\ll -\log n$.  So $p_x= o_\pr(n^{-1})$.  Similarly, $p_w=o_\pr(n^{-1})$.
     
     \item When $d$ is of order $\log n$ or higher, $a={1}/{\epsilon_d}$ with $1\ll\epsilon_d = o(\log d)$, then  $\frac{d}{2}(a\log d + O_\pr(1)) \gg d\geq O(\log n)$, and $p_x$ is also of order $o_\pr(n^{-1})$.  Similarly, $p_w=o_\pr(n^{-1})$.

		\end{enumerate}
	Under (a), (b) and (c), we all have $p_x, p_w = o_\pr(n^{-1})$. Then, $$|\pr(N_{Y_1}\in\{Y\}|Y_1=y) - \pr(N_{\tilde{Y}_1}\in \{\tilde{Y}\}|Y_1=y_w)|=|o_\pr(1)-o_\pr(1)| = o_\pr(1).$$
	
	\item  $\alpha\geq \alpha_0$:

	First we consider $\alpha_0\leq \alpha\leq \frac{\log d}{2\log n}$. By the proof of Lemma \ref{lemma:onept} and the facts that $D_{\min,w} = D_{\min,x} + o_\pr(n^{\alpha^*})$, $e^{(\kappa_1-\kappa)d/2} = \sqrt{|\Sigma|/|\Sigma_1|} = 1+o_\pr(1)$, $e^{(\kappa_2-\kappa)d/2} = \sqrt{|\Sigma|/|\Sigma_2|} = 1+o_\pr(1)$, and $\|\Sigma_2^{-1}-\Sigma_1^{-1}\|_2=o_\pr(n^{\tilde\alpha})$. Then $p_w$ is 
	 \begin{align*}
	 & p_x\left(1+\tfrac{o_\pr(n^{\alpha^*})}{O_\pr(n^{\alpha})}\right)^d e^{o_\pr\left(n^{-\frac{1}{d}-\frac{(2+a)\log d + \kappa}{2\log n}} \left(\sqrt{d}~n^\alpha+n^{2\alpha}\right)\right)}e^{o_\pr(n^{\alpha^*}\sqrt{d})} + o_\pr(n^{-1}) \\
	 & = p_x\left(1+o_\pr(n^{\alpha^*-\alpha_0})\right)^d e^{o_\pr\left(n^{-\frac{1}{d}-\frac{a\log d + \kappa}{2\log n}} \right)}e^{o_\pr(1)}+o_\pr(n^{-1}) \\
	  & = p_x(1+o(d^{-1}))^d (1+o_\pr(1))^2 +o_\pr(n^{-1}) = p_x(1+o_\pr(1))+o_\pr(n^{-1}).	
	 \end{align*}
	Hence, $p_w$ is of the same order as $p_x$ when $\alpha\geq \alpha_0$.

	 Notice that, for $Y_1, Y_2\overset{iid}{\sim} \mathcal{N}_d(\mu_1,\Sigma_1)$, $\ep(\|Y_1-Y_2\|_2) = O(\sqrt{d})$.  Thus, when $D_{\min,x} = c\sqrt{d}$ for a sufficiently large constant $c$, $p_x$ is of order $O_\pr(1)$.  Similarly, when $D_{\min,w} = c\sqrt{d}$ for a sufficiently large constant $c$, $p_w$ is of order $O_\pr(1)$.  Thus, for $\alpha>\frac{\log d}{2\log n}$, we have $D_{\min,x}, D_{\min,w} \gg O_\pr(\sqrt{d})$ and $p_x, p_w = O_\pr(1)$ are also of the same order.

	\begin{enumerate}
		\item When $p_x,p_w$ are of order $o_\pr(n^{-1})$, $|\pr(N_{Y_1}\in\{Y\}|Y_1=y) - \pr(N_{\tilde{Y}_1}\in \{\tilde{Y}\}|Y_1=y_w)|=o_\pr(1)$.
		\item When $p_x, p_w$ are of order higher than $O_\pr(n^{-1})$, the probability that no other $Y_{k^\prime} \sim \mathcal{N}_d(\mu_1,\Sigma_1)$ falls in the $D_{\min,x}$-ball of $y$ goes to 0 as $n\rightarrow\infty$, and the probability that no other $\tilde{Y}_{k^\prime}\sim \mathcal{N}_d(\mu_2,\Sigma_2)$ falls in the $D_{\min,w}$-ball of $y_w$ also goes to 0 as $n\rightarrow\infty$.  So $$|\pr(N_{Y_1}\in\{Y\}|Y_1=y) - \pr(N_{\tilde{Y}_1}\in \{\tilde{Y}\}|Y_1=y_w)| = o_\pr(1). $$
		\item When $p_x=p_w = O_\pr(n^{-1})$.  Let $\delta = |D_{\min,x}-D_{\min,w}| + \|y-y_w\|_2$.  Then $\delta = o_\pr(n^{\alpha^*})$, and $\frac{\delta d}{D_{\min,x}} = o_\pr(\frac{n^{\alpha^*}d}{n^{\alpha_0}}) = o_\pr(1)$. \\
			 Here, we define two more probabilities.  Let $p_{x,2}$ be the probability that $Y_k\sim \mathcal{N}_d(\mu_1,\Sigma_1)$ falls in the $(D_{\min,x}+\delta)$-ball of $y$, and $p_{w,2}$ be the probability that $\tilde{Y}_k\sim \mathcal{N}_d(\mu_2,\Sigma_2)$ falls in the $(D_{\min,x}+\delta)$-ball of $y$.  It is clear that both the $D_{\min,x}$-ball of $y$ and the $D_{\min,w}$-ball of $y_w$ are contained in the $(D_{\min,x}+\delta)$-ball of $y$. Because $\frac{\sqrt{d}}{D_{\min,x}} \geq O_\pr\left(\frac{\sqrt{d}}{n^{-\frac{1}{d}}d^{\frac{1-a}{2}} e^{-\frac{\kappa}{2}}}\right) = O_\pr(n^{\frac{1}{d}}d^{\frac{a}{2}}e^{\frac{\kappa}{2}}) \geq O_\pr(1)$ and $\delta\leq o_\pr(d^{-1/2})$, by the proof of Lemma \ref{lemma:onept},  we have that $p_{x,2}$ is 
			 \begin{align*}
			 	&p_x\left(1+\tfrac{\delta}{D_{\min,x}}\right)^d e^{O_\pr(\delta \sqrt{d})} + O_\pr(\delta^d d^{-d/2}e^{\kappa_1 d/2})\\
			 	& = p_x~e^{d~O_\pr\left(\frac{\delta}{D_{\min,x}}\right)}e^{O_\pr(\delta \sqrt{d})} + o_\pr(\tfrac{\delta d}{D_{\min,x}} n^{-1}) \\
			 	& = p_x\left(1+O_\pr\left(\tfrac{\delta d}{D_{\min,x}}\right)\right) = p_x(1+o_\pr(1)),
			 \end{align*}
			 Similarly, $p_{w,2} = p_w(1+O_\pr(\frac{\delta d}{D_{\min,x}})) = p_w(1+o_\pr(1)).$
Then $p_{x,2}$ and $p_{w,2}$ are also of order $O_\pr(n^{-1})$.  \\
			Based on the proof of Lemma \ref{lemma:onept_ori}, $p_{x,2}$ and $p_{w,2}$ differ by a factor of 
			\[
			1 + O_\pr(dn^{\tilde \alpha}) = 1 + O_\pr(d^{-\frac{a}{2}-\frac{\log n}{d\log d}-O(\frac{1}{\log d})}) = 1 + o_\pr(1).
			\]
			
			Notice that $p_{x,2}-p_x = p_x(p_{x,2}/p_x-1) = O_\pr(n^{-1}) O_\pr(\frac{\delta d}{D_{\min,x}})=O_\pr(n^{-1} \delta d/D_{\min,x})$.  Similarly, $p_{w,2}-p_w = O_\pr(n^{-1} \delta d/D_{\min,x})$.  \\
			Let $\xi_{\min}=\min\{\delta d/D_{\min,x}, d^{-a-\frac{\log n}{d\log d}-O(\frac{1}{\log d})}\}$, we have that $|p_x-p_w| = O_\pr(n^{-1} \xi_{\min}).$  Let $p_x = c_0 n^{-1}$.  Then $p_w = c_0 n^{-1} + c_1 n^{-1} \xi_{\min} + o_\pr(n^{-1} \xi_{\min})$ for a constant $c_1$.  Then, 
\begin{align*}
&|\pr(N_{Y_1}\in\{Y\}|Y_1=y) - \pr(N_{\tilde{Y}_1}\in \{\tilde{Y}\}|\tilde{Y}_1=y_w)| \\
&	 = |(1-(1-p_x)^{n-1}) - (1-(1-p_w)^{n-1})| \\ 
&	 = |(1-p_w)^{n-1} - (1-p_x)^{n-1}| \\ 
&	 = |(1-c_0 n^{-1})^{n-1} - (1-c_0 n^{-1} - c_1 n^{-1}\xi_{\min}-o_\pr(n^{-1}\xi_{\min}))^{n-1}| \\
&	 = |(n-1) (1-c_0 n^{-1})^{n-2} c_1 n^{-1}\xi_{\min} + o_\pr(1)| \\
&	 = o_\pr(1).
\end{align*}			
			 
	\end{enumerate}
		
	 \end{enumerate}
	 
	 Thus, under all possibilities of scenarios (1) and (2), we have $|\pr(N_{Y_1}\in\{Y\}|Y_1=y) - \pr(N_{\tilde{Y}_1}\in \{\tilde{Y}\}|\tilde{Y}_1=y_w)| = o_\pr(1)$.  Hence, 
	 \begin{align*}
	 & |P(YY)-P(Y^*Y^*)| \\ 
	 & \leq \int_{x_1,\dots,x_n} \int_y |\pr(N_{Y_1}\in\{Y\}|Y_1=y) - \pr(N_{\tilde{Y}_1}\in \{\tilde{Y}\}|\tilde{Y}_1=y_w))| \\
	 & \hspace{75mm} \times g_1(y) dy \prod_{i=1}^n f(x_i) dx_i \\
	 & = o(1),	 	
	 \end{align*}
	 and the conclusion of the theorem follows.
		
%
%
%

\section{Technical Lemmas}
\label{sec:lemma}

\begin{lemma}\label{lemma:square_root}
For independent observations $X_1,\ldots,X_n\overset{iid}{\sim} \mathcal{N}_d(\mu,\S)$,  assume that  $\lambda_{\min}(\S)\geq C$ for some constant $C>0$. If 
$\|\Sigma_x-\Sigma\|_2 = o_{\pr}(r_{n,d})$ with $r_{n,d}=O(1)$, then we have $\|\Sigma_x^{1/2}-\Sigma^{1/2}\|_2 = o_{\pr}(r_{n,d})$.
\end{lemma}

\begin{proof}
Denote by $v\in R^d$ an eigenvector of $\Sigma_x^{1/2}-\Sigma^{1/2}$ of unit length, we have
\beas
|(\S_x^{1/2}v-\S^{1/2}v)^T(\S_x^{1/2}v+\S^{1/2}v)|=|v^T(\S_x-\S)v|\leq \|\Sigma_x-\Sigma\|_2 = o_{\pr}(r_{n,d}).
\eeas
Suppose that $(\S_x^{1/2}-\S^{1/2})v=\lambda v$, then we have that
\[
|\lambda v^T(\S_x^{1/2}+\S^{1/2})v|= o_{\pr}(r_{n,d}).
\]
By the condition that $\lambda_{\min}(\S)\geq C$ and that $\|\Sigma_x-\Sigma\|_2 = o_{\pr}(r_{n,d})$, we have
\[
v^T\S_x^{1/2}v = v^T\S^{1/2}v + v^T(\S_x^{1/2}-\S^{1/2})v \geq C-o(1)
\]
with probability going to 1. Hence, for some constant $C_0>0$, we have, with probability tending to 1,
\[
v^T(\S_x^{1/2}+\S^{1/2})v\geq C_0.
\]
It yields that $\lambda=o_{\pr}(r_{n,d})$. Since $v$ could be any eigenvector of $\Sigma_x^{1/2}-\Sigma^{1/2}$, we have
\[
\|\Sigma_x^{1/2}-\Sigma^{1/2}\|_2 = o_{\pr}(r_{n,d}).
\]
\end{proof}

\begin{lemma}\label{lemma:xdist}
{ Suppose $x \sim  \mathcal{N}_d(0,\Sigma)$.} Then under the conditions of Theorem \ref{thm:asym}, we have
\begin{align}\label{eq:error}
 \|\Sigma_x^{1/2}\Sigma^{-1/2}x - x\|_2 = o_{\pr}(n^{-\frac{1}{d}-\frac{(1+a)\log d + \kappa}{2 \log n}}).
\end{align}
\end{lemma}

\begin{proof}
\beas
\|\Sigma_x^{1/2}\Sigma^{-1/2}x - x\|_2 &=& \|\Sigma_x^{1/2}\Sigma^{-1/2}x-\Sigma^{1/2}\Sigma^{-1/2}x\|_2\cr
&=& \|(\Sigma_x^{1/2}-\Sigma^{1/2})\Sigma^{-1/2}x\|_2.
\eeas

Let $z=\Sigma^{-1/2} x$, we have
\beas
\|\Sigma_x^{1/2}\Sigma^{-1/2}x - x\|_2 &\leq& \|\Sigma_x^{1/2}-\Sigma^{1/2}\|_2\|z\|_2.
\eeas

Notice that the covariance matrix of $z$ is an identity matrix, $\|z\|_2^2/d$ converges to a constant almost surely { as $d\rightarrow \infty$.}
By the condition  that $\|\Sigma_x^{1/2}-\Sigma^{1/2}\|_2 = o_{\pr}(n^{-\frac{1}{d}-\frac{(2+a)\log d + \kappa}{2 \log n}})$, we have that  
\[
\|\Sigma_x^{1/2}\Sigma^{-1/2}x - x\|_2\leq o_{\pr}(n^{-\frac{1}{d}-\frac{(2+a)\log d + \kappa}{2 \log n}})\|z\|_2=o_{\pr}(n^{-\frac{1}{d}-\frac{(1+a)\log d + \kappa}{2 \log n}}).
\]

\end{proof}

\begin{lemma}\label{lemma:onept_ori}
Let $X_1 \sim \mathcal{N}_d(\mu,\Sigma)$, $Y$ independent of $X_1$'s and $Y\sim \mathcal{N}_d(\mu_x,\Sigma_x)$, where $\mu$, $\Sigma$, $\mu_x$, and $\Sigma_x$ satisfy the conditions in Theorem \ref{thm:asym}. 
\begin{enumerate}
\item When $d$ is fixed, for $r=o(1)$, the probability $Y$ falls in the $r$-ball centered at $X_1$ is of order $O_\pr(r^d)$. 
\item When $d$ increases with $n$, for $r=O(d^\beta)$, $\beta\leq \frac{1}{2}$, the logarithm of the probability $Y$ falls in the $r$-ball centered at $X_1$ is $-\frac{1}{2} d\log d + d \log r + \frac{1}{2} d (\kappa+O_\pr(1))$.  More specifically, when $\beta\leq -0.5$, the probability $Y$ falls in the $r$-ball centered at $X_1$ is of order $O_\pr(r^d d^{-d/2} e^{\kappa d/2})$. 
	
\end{enumerate}

\end{lemma}
\begin{proof}
Under a special case that $\mu_x=\mathbf{0}$, $\Sigma_x=I$ and $X_1\equiv\mathbf{0}$, 
 the probability is
\begin{align*}
& \int_0^r \frac{d \pi^{d/2}}{\Gamma(d/2+1)} t^{d-1} \frac{1}{(2\pi)^{d/2}} e^{-\frac{1}{2}t^2} dt  = \frac{d}{2^{d/2}\Gamma(d/2+1)}\int_0^r  t^{d-1} e^{-\frac{1}{2}t^2} dt,
\end{align*}
which is of order $2^{-d/2} e^{-d/2 \log (d/2) + d/2}\int_0^r t^{d-1} e^{-\frac{1}{2}t^2} dt = d^{1-d/2} e^{d/2}\int_0^r t^{d-1} e^{-\frac{1}{2}t^2} dt$.

For generic $\mu_x, \Sigma_x$ and $X_1\sim \mathcal{N}_d(\mu,\Sigma)$, notice that 
\begin{align*}
\frac{f_{\mu_x,\Sigma_x}(Y)}{f_{\mu_x,\Sigma_x}(X_1)} & = \frac{e^{-\frac{1}{2}(Y-\mu_x)^T\Sigma_x^{-1}(Y-\mu_x)}}{e^{-\frac{1}{2}(X_1-\mu_x)^T\Sigma_x^{-1}(X_1-\mu_x)}} \\ 
& = e^{(X_1-Y)^T\Sigma_x^{-1}(X_1-\mu_x) - \frac{1}{2}(X_1-Y)^T\Sigma_x^{-1}(X_1-Y)}.
\end{align*}
When $\|Y-X_1\|_2 = t$, based on the conditions in Theorem \ref{thm:asym}, there exists a positive function $c_1(t)$ and a positive constant $c_2$ such that  $\frac{f_{\mu_x,\Sigma_x}(Y)}{f_{\mu_x,\Sigma_x}(X_1)} = O_\pr(e^{c_1(t) t \sqrt{d} - c_2 t^2})$.  Then, probability $Y$ falls in the $r$-ball of $X_i$ is of order
\begin{align*}
& d^{1-d/2} e^{d/2}\int_0^r t^{d-1} e^{-\frac{1}{2}t^2} \int |\Sigma_x|^{-\frac{1}{2}} e^{-\frac{1}{2}(x-\mu_x)^T\Sigma_x^{-1}(x-\mu_x)} \\
& \hspace{40mm} \times (2\pi)^{-\frac{d}{2}} |\Sigma|^{-\frac{1}{2}} e^{-\frac{1}{2}(x-\mu)^T\Sigma^{-1}(x-\mu)}	 e^{c_1(t) t \sqrt{d} - c_2 t^2} dx dt.	
\end{align*}
Under the conditions of Theorem \ref{thm:asym}, we have 
\begin{align*}
&	\int |\Sigma_x|^{-\frac{1}{2}} e^{-\frac{1}{2}(x-\mu_x)^T\Sigma_x^{-1}(x-\mu_x)}(2\pi)^{-\frac{d}{2}} |\Sigma|^{-\frac{1}{2}} e^{-\frac{1}{2}(x-\mu)^T\Sigma^{-1}(x-\mu)} dx \\
& = O\left(|\Sigma|^{-\frac{1}{2}}|\Sigma_x|^{-\frac{1}{2}}|\Sigma^{-1}+\Sigma_x^{-1}|^{-\frac{1}{2}} \right) = O(2^{-d/2}|\Sigma|^{-\frac{1}{2}}).
\end{align*}
Thus, the probability $Y$ falls in the $r$-ball of $X_i$ is of order 
\begin{align}\label{eq:prob_r}
	d^{1-d/2} e^{\kappa d/2} \int_0^r t^{d-1} e^{c_1(t) t\sqrt{d} - (c_2+0.5) t^2} dt
\end{align}


 We first consider the cases when $d$ increases with $n$.  Suppose $t=c_0 d^\beta$, for some fixed $\beta$ and $0<c_0\leq C$ for some constant $C>0$. Then the integrand is $$e^{(d-1)(\beta\log d + \log c_0) + c_1(t) c_0 d^{\beta+0.5} - (c_2+0.5) c_0^2 d^{2\beta}}.$$ We consider the following two scenarios.
\begin{enumerate}[(1)]
\item
If $\beta<0$ or $0<\beta\leq \frac{1}{2}$, then $d\log d$ dominates the other terms. Furthermore, we have that 
\begin{align*}
& \frac{d\log d}{c_0 d^{\beta+0.5}} = \frac{1}{c_0} d^{0.5-\beta}\log d \geq O(\log d), \\
& \frac{d\log d}{c_0^2 d^{2\beta}} = \frac{1}{c_0^2} d^{1-2\beta} \log d \geq O(\log d).
\end{align*}

\item
If $\beta=0$, we further consider the following two cases.  Let $\epsilon_1 = \frac{\log d}{\sqrt{d}}$,
	\begin{enumerate}
	\item
	if $c_0\geq 1+\epsilon_1$, then $d\log c_0$ dominates the other terms, and we have that
	\[
	\frac{d\log c_0}{c_0 d^{0.5}}  \geq O(\log d), \mbox{ and } \frac{d\log c_0}{c_0^2} \geq O(\log d \sqrt{d}) \gg O(\log d).
	\]
	\item
	if $c_0\leq 1-\epsilon_1$, again $d\log c_0$ dominates the other terms, and we have that
	\[
	\Big{|}\frac{d\log c_0}{c_0 d^{0.5}}\Big{|} \geq O(\log d), \mbox{ and } \Big{|}\frac{d\log c_0}{c_0^2}\Big{|} \geq O(\log d \sqrt{d}) \gg O(\log d).
	\]
	\end{enumerate}
\end{enumerate}
First of all, when $\beta\leq -0.5$, from scenario (1), we have 
\[
\frac{d\log d}{c_0 d^{\beta+0.5}} \geq O(d\log d), \mbox{ and } \frac{d\log d}{c_0^2 d^{2\beta}}  \geq O(d^2\log d).
\]
Then,
\begin{align*}
& d^{1-d/2} e^{\kappa d/2} \int_0^r t^{d-1} e^{c_1(t) t\sqrt{d} - (c_2+0.5) t^2} dt \\
&= d^{1-d/2} e^{\kappa d/2} \int_0^r e^{(d-1)\log t(1+ O(\frac{1}{d\log d}))} dt \\
&= d^{-d/2} e^{\kappa d/2} r^d r^{O(\frac{1}{\log d})} = O(d^{-d/2} r^d e^{\kappa d/2}).
\end{align*}
When $-0.5<\beta\leq 0.5$, based on scenarios (1) and (2), we have
\begin{align}
& d^{1-d/2} e^{\kappa d/2} \int_0^r t^{d-1} e^{c_1(t) t\sqrt{d} - (c_2+0.5) t^2} dt \nonumber \\
&= d^{1-d/2} e^{\kappa d/2} \int_0^r e^{(d-1)\log t(1+ O(1/\log d))} dt \label{eq:around1} \\
&=d^{-d/2} e^{\kappa d/2} r^d r^{O(d/\log d)} = d^{-d/2}r^d e^{\frac{1}{2}d(\kappa+O(\frac{\log r}{\log d}))} \nonumber \\
&= d^{-d/2}r^d e^{\frac{1}{2}d(\kappa+O(1))}. \nonumber
\end{align}
For \eqref{eq:around1}, the part of the integral from $1-\epsilon_1$ to $1+\epsilon_1$ is not an issue: Notice that $\int_{1-\epsilon_1}^{1+\epsilon_1} d t^{d-1} dt = (1+\epsilon_1)^d - (1-\epsilon_1)^d = e^{d \log (1+\epsilon_1)} - e^{d \log (1-\epsilon_1)} = e^{d O(\epsilon_1)}  = e^{O(\sqrt{d} \log d)} $, and $\int_{1-\epsilon_1}^{1+\epsilon_1} d e^{c^* t\sqrt{d}} dt = \frac{\sqrt{d}}{c^*} e^{c^* \sqrt{d}}(e^{1+\epsilon_1}-e^{1-\epsilon_1}) = O(\epsilon_1 \sqrt{d} e^{c^* \sqrt{d}} ) = O(\log d e^{c^* \sqrt{d}}) = e^{O(\sqrt{d})}$ with $c^* = \sup_{1-\epsilon_1 \leq t \leq 1+\epsilon_1} c_1(t)$ a positive constant.  Then, the difference between the two integrals is at most $e^{O(\sqrt{d} \log d)}$, which is much smaller than $e^{O(d)}$ and thus does not affect the above result.


When $d$ is fixed, the proofs are much simpler, and it is not hard to see that, when $r=o(1)$, the probability is of order $r^d$. 
\end{proof}

\begin{lemma}\label{lemma:onept}
Let $Y_1, Y_2 \overset{iid}{\sim} \mathcal{N}_d(\mu_0,\Sigma_0)$, where $\|\mu_0\|_\infty$ is bounded by a positive constant, and $\|\Sigma_0^{-1}-\Sigma^{-1}\|_2 = o(1)$.  
\begin{enumerate}
\item When $d$ is fixed, for $r=o(1)$, the probability $Y_2$ falls in the $r$-ball centered at $Y_1$ is of order $O_\pr(r^d)$. 
\item When $d$ increases with $n$, for $r=O(d^\beta)$, $\beta\leq \frac{1}{2}$, the logarithm of the probability $Y_2$ falls in the $r$-ball centered at $Y_1$ is $-\frac{1}{2} d\log d + d \log r + \frac{1}{2} d (\kappa_0+O_\pr(1))$, where $\kappa_0 = 1 - \frac{\log|\Sigma_0|}{d}-\log 2$.  More specifically, when $\beta\leq -0.5$, the probability $Y_2$ falls in the $r$-ball centered at $Y_1$ is of order $O_\pr(r^d d^{-d/2} e^{\kappa_0 d/2})$. 
\end{enumerate}

%
\end{lemma}

\begin{proof}
	Based on the proof of Lemma \ref{lemma:onept_ori}, the probability $Y_2$ falls in the $r$-ball of $Y_1$ is of order 
	\begin{align*}
& d^{1-d/2} e^{d/2}\int_0^r t^{d-1} e^{-\frac{1}{2}t^2} \int |\Sigma_0|^{-\frac{1}{2}} e^{-\frac{1}{2}(x-\mu_0)^T\Sigma_0^{-1}(x-\mu_0)}(2\pi)^{-\frac{d}{2}} |\Sigma_0|^{-\frac{1}{2}}\\
&\hspace{37mm} \times e^{-\frac{1}{2}(x-\mu_0)^T\Sigma_0^{-1}(x-\mu_0)}	 e^{c_1(t) t \sqrt{d}(1+o(1)) - c_2 t^2(1+o(1))} dx dt \\
& = d^{1-d/2} e^{d/2} 2^{-d/2} |\Sigma_0|^{-\frac{1}{2}} \int_0^r t^{d-1} e^{c_1(t) t\sqrt{d}(1+o(1)) - (c_2+0.5) t^2(1+o(1))} dt \\
& = d^{1-d/2} e^{\kappa_0 d/2} \int_0^r t^{d-1} e^{c_1(t) t\sqrt{d}(1+o(1)) - (c_2+0.5) t^2(1+o(1))} dt,
\end{align*}
with $c_1(t)$ and $c_2$ the same as those in the proof of Lemma \ref{lemma:onept_ori}. Then, with the same arguments as in the proof of Lemma \ref{lemma:onept_ori}, the results of this lemma follow.

\end{proof}


\section{Acknowledgements}
{ The authors thank the Associate Editor and two referees for their helpful constructive comments which have helped to improve the quality and presentation of the paper. The authors also thank the reviewer's suggestions on various numerical alternatives including the improved multivariate Shapiro--Wilk's test and the Fisher's test.}

\bibliographystyle{apalike}
\bibliography{test}


\end{document}